\documentclass[a4paper,11pt,reqno]{amsart}
\usepackage{etex}
\usepackage{amssymb}
\usepackage{amsmath}
\usepackage{amsfonts}
\usepackage{amsthm}
\usepackage{latexsym}
\usepackage{mathrsfs}
\usepackage{eucal}
\usepackage[safe]{tipa}
\usepackage{cases}

\usepackage[width=21cm, left=1in, right=1in,marginpar=2cm,headsep=1cm,footskip=1cm]{geometry}

\usepackage{xcolor}

\usepackage{slashed}
\usepackage{textalpha}

\usepackage{lpic}

\usepackage{tikz}
\usetikzlibrary{cd}
\usetikzlibrary{arrows} 
\usetikzlibrary{graphs}
\usetikzlibrary{intersections}
\usetikzlibrary{positioning}
\usetikzlibrary{decorations.pathmorphing}
\usetikzlibrary{datavisualization} 
\usetikzlibrary{datavisualization.formats.functions}
\usetikzlibrary{intersections}
\usetikzlibrary{hobby}
\usetikzlibrary{through}
\usetikzlibrary{calc}
\usepackage{pgfmath}
\usepackage{pgfkeys}

\usetikzlibrary{automata,positioning}
\usetikzlibrary{decorations.markings}
\tikzset{->-/.style={decoration={
  markings,
  mark=at position #1 with {\arrow{>}}},postaction={decorate}}}
  \tikzset{-<-/.style={decoration={
  markings,
  mark=at position #1 with {\arrowreversed[red]{latex'}}},postaction={decorate}}}

\usepackage{url}
\usepackage[linktocpage=true]{hyperref}

\hypersetup{
    colorlinks=true,
    linkcolor=blue,
    filecolor=blue,      
    urlcolor=blue,
}

\newcommand{\Reals}{\mathbb{R}}

\newcommand{\half}{\tfrac{1}{2}}

\newcommand{\cM}{\mathcal{M}}
\newcommand{\Domain}{\Omega}

\newcommand{\Rod}{\mathcal{R}}
\newcommand{\Nut}{n}

\newcommand{\xiInt}{\mathcal{I}}

\theoremstyle{plain}
\newtheorem{thm}{Theorem}[section]
\newtheorem{cor}[thm]{Corollary}
\newtheorem{lemma}[thm]{Lemma}

\newtheorem{definition}[thm]{Definition}
\newtheorem{prop}[thm]{Proposition}
\newtheorem{conj}{Conjecture} 

\newtheorem{question}[thm]{Question}
\newtheorem{remark}[thm]{Remark}

\renewcommand{\d}{\mathrm{d}}

\DeclareMathOperator{\SL}{\rm SL}
\DeclareMathOperator{\CP}{\mathbb{C}\mathbb{P}}

\newcommand{\ba}{\mathbf{a}}

\newcommand{\SU}{\text{SU}}

\title{Gravitational Instantons and special geometry} 
\author[S. Aksteiner]{Steffen Aksteiner}
\email{steffen.aksteiner@aei.mpg.de}
\address{Albert Einstein Institute, Am M\"uhlenberg 1, D-14476 Potsdam, Germany }
\author[L. Andersson]{Lars Andersson}
\email{laan@aei.mpg.de}
\address{Albert Einstein Institute, Am M\"uhlenberg 1, D-14476 Potsdam, Germany }

\allowdisplaybreaks[2]

\numberwithin{equation}{section}

\begin{document}

\begin{abstract}
The Chen-Teo gravitational instanton is an asymptotically flat, toric, Ricci flat family of metrics on $\CP^2\setminus S^1$, that provides a counterexample to the classical Euclidean Black Hole Uniqueness conjecture. In this paper we show that the Chen-Teo instanton is Hermitian and non-Kähler. It follows that all known examples of gravitational instantons are Hermitian. 
\end{abstract} 

\maketitle

\section{Introduction}
A gravitational instanton
is a complete, non-compact, Riemannian four-manifold with at least quadratic curvature decay and vanishing Ricci tensor. A set of uniqueness conjectures for gravitational instantons were formulated in the 1970's \cite{gibbons:1978:survey}, and much work has since then been devoted to the problems of constructing and classifying instantons. In addition to the asymptotically locally Euclidean (ALE) and asymptotically locally flat (ALF) cases, with quartic and cubic volume growth, respectively, examples with quadratic and slower volume growth have been found, and a classification has been proved for hyperk\"ahler instantons, cf. \cite{chen2015gravitational,chen2019gravitational,chen2021gravitational,2021arXiv210812991S} and references therein.   For an ALF instanton, the boundary at infinity is a circle bundle over $S^2$, and the case when this bundle is trivial is called asymptotically flat (AF). 

The 2-parameter family Euclidean Kerr instantons on $S^4 \setminus S^1$ \cite[\S II.B]{1998PhRvD..59b4009H} is the most well known example of an AF instanton, and  
according to the classical form of the Euclidean Black Hole Uniqueness Conjecture  \cite[Conjecture 2]{gibbons:1978:survey}, a non-flat AF gravitational instanton is in the Kerr family. However, the remarkable instanton found 
by Chen and Teo, and discussed in detail in Sections  \ref{sec:CTmetric} and \ref{sec:CTInst} below, shows that the conjecture is not valid without additional assumptions. 
\begin{thm}[Chen and Teo \cite{2011PhLB..703..359C}]\label{thm:CT}
Let $\cM = \CP^2 \setminus S^1$. There is a 2-parameter family of AF instantons $(\cM, g_{ab})$ that admit an effective isometry action by the 2-torus, with three isolated fixed points.   
\end{thm} 

Let $(\cM, g_{ab})$ be an oriented Riemannian four-manifold. The Weyl tensor decomposes into self dual and anti-self dual parts that define self-adjoint endomorphisms, with vanishing trace, on the 3-dimensional spaces of (anti)-self dual 2-forms. Using terminology originating in the Petrov classification of algebraic types  \cite{karlhede:1986,PR:II}, the cases with exactly three, two, or one distinct eigenvalues are called type $I, D$ or $O$, respectively. Thus, for example the Euclidean Kerr instanton is of type $D^+ D^-$, where we use $+,-$ to indicate self dual and anti-self dual parts, respectively. A Weyl tensor of type $I$ is also called algebraically general, while types $D$ and $O$ are called algebraically special. We are now ready to state our main theorem. 
\begin{thm} \label{thm:main-intro} 
The 2-parameter Chen-Teo family of  AF gravitational instantons 
is of type $D^+ I^-$.
\end{thm}
The proof is given in Section \ref{sec:ProofMainThm}. The theorem shows that the Chen-Teo instanton is one-sided type $D$ in the sense of \cite{1984GReGr..16..797P,2020arXiv200303234T}, which has several important consequences. If $(\cM, g_{ab})$ is Ricci flat and one-sided type $D$, then it is Hermitian and conformally Kähler \cite{MR707181, 1984GReGr..16..797P,MR757212} and, since a Ricci flat K\"ahler four-manifold is half flat, Theorem~\ref{thm:main-intro} implies that the Chen-Teo instanton is non-Kähler.
In addition, the metric can be represented in terms of a potential solving the $\SU(\infty)$ Toda equation, cf. \cite{1984GReGr..16..797P,2020arXiv200303234T} and references therein. 
The Chen-Teo family of gravitational instantons is contained in the 5-parameter Chen-Teo family of Ricci flat metrics \cite{2015PhRvD..91l4005C}, which in general are not smooth. We prove in this paper that  for generic values of the parameters, also the Chen-Teo 5-parameter family  is $D^+ I^-$, cf.  Theorem \ref{thm:AlgNonSpec}. Moreover, the Weyl tensor is of type $D^+ D^-$ at points where at least one of the Killing fields vanishes.  

As mentioned above, the Euclidean Kerr instanton is the Wick-rotated version of the Lorentz\-ian Kerr black hole metric. In other words, the complexification of the Euclidean Kerr instanton contains a regular Lorentzian real section. 
However, a Ricci flat Riemannian four-manifold has this property only if the self dual and anti-self dual parts of the Weyl tensor have the same algebraic type, cf. \cite{1977RpMP...11..197R}, see also \cite{1977IJTP...16..663W}. This yields the following corollary to Theorem \ref{thm:main-intro} that answers a question raised in \cite{2011PhLB..703..359C,2015PhRvD..91l4005C}, see also \cite[p. 9]{2018JHEP...02..008B} for an argument based on Ernst potentials.
\begin{cor}\label{cor:no-lor-intro} 
Let $(\cM,g_{ab})$ be in the Chen-Teo family of gravitational instantons. The complexification of $(\cM, g_{ab})$ does not contain a  real Lorentzian section. 
\end{cor}

In view of Theorem \ref{thm:main-intro}, all known examples of gravitational instantons are Hermitian. 
As mentioned above, there is a classification of hyperk\"ahler instantons. Further, there is a classification of compact Hermitian-Einstein four-manifolds \cite{MR2899877}. 
An important step in the proof of the just mentioned classification result is to show that the symmetry group of  compact Hermitian-Einstein four-manifolds contains a 2-torus. 
This holds for all known Hermitan, non-K\"ahler gravitational instantons, i.e. Euclidean Kerr, Chen-Teo, Taub-bolt, or Taub-NUT with the orientation opposite to the hyperk\"ahler orientation. Further, it has recently been proved by Biquard and Gauduchon \cite{Biquard:private} that a Hermitan, non-K\"ahler ALF instanton whose isometry group contains a 2-torus, must be on the just mentioned list of examples. This motivates the following conjecture.
\begin{conj} \label{conf:Herm} 
Let $(\cM, g_{ab})$ be a Hermitian, non-K\"ahler, ALF gravitational instanton. Then  $(\cM, g_{ab})$ is one of  Euclidean Kerr, Chen-Teo,  Taub-bolt, or Taub-NUT with the orientation opposite to the hyperk\"ahler orientation.
\end{conj} 

Recall that all known examples of compact Ricci flat spaces have special holonomy. The conjecture that this is true for any compact Ricci-flat space has been called the Besse conjecture, see \cite[p. 19]{MR2371700}, \cite[Remark 1.2]{botvinnik1999rigidly}. Although attempts have been made to construct compact Ricci flat four-manifolds that do not have special holonomy, cf. \cite{brendle2017gluing} and references therein, these have not been successful. The analog of the Besse conjecture for gravitational instantons would be the statement that all gravitational instantons are Hermitian. We therefore ask the following 
\begin{question} 
Are there non-Hermitian gravitational instantons? 
\end{question}

\subsection*{Overview of this paper} 
In Section~\ref{sec:prel} we provide some background material and notation. The Chen-Teo 5-parameter family of metrics is introduced in Section~\ref{sec:CTmetric} and its rod structure is described in Section~\ref{sec:rodstruct}. Section \ref{sec:AlgSpec} proves that the Chen-Teo metric is algebraically special. The restriction from the Chen-Teo 5-parameter family of metrics to the Chen-Teo 2-parameter instanton family is described in Section~\ref{sec:CTInst} and regularity and asymptotic flatness is proved in Sections~\ref{sec:regularity}  and \ref{sec:asymptotics}, respectively. Section~\ref{sec:ProofMainThm} is devoted to the proof of our main theorem.

\section{Preliminaries and notation} \label{sec:prel}

We shall use abstract index notation following \cite{PR:I,PR:II}. Let $(\cM, g_{ab}, \nabla_a)$ be a Riemannian four-manifold with metric $g_{ab}$ and Levi-Civita connection $\nabla_a$. Most computations were performed using a modified version of the package Spinframes \cite{SpinFrames} allowing for Riemannian signature. The package is  based on the symbolic computer algebra package xAct for Mathematica\textsuperscript{TM}.

\subsection{AF instantons} \label{sec:AF} 
Since we shall mainly be concerned with AF geometry in this paper, we give a precise definition.  Recall that it follows from the Cheeger-Gromoll splitting theorem that a non-flat gravitational instanton has one end.  
\begin{definition}[AF Instanton, \protect{\cite[\S II]{lapedes:1980}}] \label{def:AF}
Let $\kappa, \Omega \in \Reals$, $\kappa \ne 0$. Consider $\Reals^4$ with coordinates $(\tau, r, \theta, \phi) \in \Reals \times \Reals_+ \times [0,\pi] \times [0,2\pi]$ so that the flat metric takes the form 
\begin{align} 
\d\tau^2 + \d r^2 + r^2(\d\theta^2 + \sin^2\theta \d\phi^2) .
\end{align} 
Let $(\mathring{\cM}, \mathring{g}_{ab})$ be the flat space defined as $\Reals^4/\sim$ where the equivalence relation $\sim$ is given by the identification 
\begin{align}
(\tau, r,\theta,\phi) \sim (\tau+2\pi/\kappa,r,\theta,\phi + 2\pi\Omega/\kappa) .
\end{align} 
A non-flat Riemannian instanton $(\cM, g_{ab})$ is said to be AF, with parameters $\kappa, \Omega$, if there is a compact $K \subset \cM$ and a diffeomorphism $\Phi: \cM \setminus K \to \mathring{\cM} \setminus \{ r \leq 1\}$ such that 
\begin{align} 
g - \Phi^* \mathring{g} = O(1/r) .
\end{align}  
\end{definition} 
\begin{remark}
The Euclidean Kerr instanton is AF with parameters $\kappa, \Omega$ corresponding to surface gravity and rotation speed, respectively. The non-rotating Euclidean Schwarzschild instanton is the limit of Kerr with $\Omega = 0$.  
\end{remark}

\subsection{Rod structures and regularity} \label{sec:rods}

The Chen-Teo instanton belongs to a 5-parameter family of metrics \cite{2015PhRvD..91l4005C} which in general have conical singularities. It was constructed using the Belinski-Zakharov  \cite{1978JETP...48..985B,1979JETP...50....1B} soliton method. This leads to a  metric with two commuting Killing fields, in Weyl-Papapetrou coordinates $(x^a) = (\tau, \phi, \rho, z)$,  
\begin{align} \label{eq:WP-metric} 
g_{ab} \d x^a \d x^b = G_{ij} \d\varphi^i \d\varphi^i + \mathbf{f}(\d\rho^2 + \d z^2)
\end{align} 
Here $(\varphi^1,\varphi^2)= (\tau, \phi)$ are Killing coordinates,  with corresponding Killing fields $\partial_\tau$, $\partial_\phi$. 
The metric coefficients $G_{ij}, \mathbf{f}$ are functions of $(\rho, z) \in \Reals_+ \times \Reals$, with $\rho$ related to the Gram matrix $G = (G_{ij})_{i,j=1,2}$ by  
\begin{align} \label{eq:rhoG} 
\rho^2 = \vert \!\det G\vert .
\end{align} 
The four-metric $g_{ab}$ is positive definite for $\rho > 0$ provided $G_{ij}$ is positive definite and $\mathbf{f} > 0$. However, the metric constructed by the Belinski-Zakharov method leads to a  smooth four-manifold (after suitable identifications of the Killing coordinates $\tau,\phi$) only if additional conditions on the behavior of the metric near the axis $\rho = 0$ are imposed.

From \eqref{eq:rhoG}, the Gram matrix has a non-trivial kernel on the axis $\rho = 0$ and  vanishes for isolated values  $z_1 < z_2< \cdots< z_n$ of $z$. These points, called turning points, are zeros for both Killing fields $\partial_\tau, \partial_\phi$. Letting $z_0 = -\infty$, $z_{n+1} = \infty$, the intervals $(z_k, z_{k+1})$, $k \in \{0,\dots, n\}$ are called rods. For each rod, the kernel of the Gram matrix is one-dimensional, defining a rod-vector $K_k$, with surface gravity
\begin{align} 
\kappa_k = \sqrt{- \half \nabla_a (K_k)_b \nabla^a (K_k)^b} \bigg{|}_{\text{Rod $k$}}. 
\end{align} 
where the limit $\rho \to 0$ is taken on each rod. In order to construct a smooth 4-manifold, the Killing coordinates must be identified in such a way that the one-parameter groups generated by each  normalized rod vector 
\begin{align} \label{eq:RodkGeneral}
\ell_k = K_k/\kappa_k = \ell_k^\tau \partial_\tau + \ell_k^\phi \partial_\phi 
\end{align} 
have closed orbits, with period $2\pi$. This is possible if and only if the $2\times 2$ matrices $\Phi_k$ defined by 
\begin{align} 
(\ell_k,\ell_{k+1}) = \Phi_k (\ell_{k-1},\ell_{k}), \quad k = 1,\dots,n , 
\end{align} 
and their inverses $\Phi_k^{-1}$ have integer entries, i.e. $\Phi_k \in \SL(2,\mathbb{Z})$. In particular, 
\begin{align}\label{eq:RegCond}  
\det \Phi_k = \pm 1, \quad k=1,\dots, n.
\end{align} 
Provided that $g_{ab}$ is smooth and non-degenerate for $\rho > 0$, and that these  regularity conditions are satisfied, making the identifications
\begin{align}\label{eq:tauphi-identify}  
(\tau,\phi) \sim (\tau,\phi) + 2\pi (\ell_0^\tau,  \ell_0^\phi) , \quad
(\tau,\phi) \sim (\tau,\phi) + 2\pi (\ell_1^\tau,  \ell_1^\phi),
\end{align}
the metric extends smoothly to the axis $\rho = 0$, and we have a smooth,  Ricci-flat, four-manifold $(\cM, g_{ab})$. By construction, the torus group $U(1) \times U(1)$ acts effectively by isometries on $(\cM, g_{ab})$, with $n$ fixed points, corresponding to the turning points, cf. \cite{2010NuPhB.838..207C}, \cite{2016PhRvD..93d4021C}, \cite{Nilsson-Inst}.
\begin{remark}
With component notation $\ell_k = (\ell_k^\tau,\ell_k^\phi)$ for the vector $\ell_k^\tau \partial_\tau + \ell_k^\phi \partial_\phi$, the normalized rod vectors \eqref{eq:RodkGeneral} take the form \cite{2015PhRvD..91l4005C}
\begin{align}\label{eq:RodkGram}
\ell_k = \pm \left( \sqrt{\mathbf{f} g_{\phi\phi}}, \sqrt{\mathbf{f} g_{\tau\tau}} \right)\bigg{|}_{\text{Rod $k$}},
\end{align}
with $\mathbf{f}$ given in \eqref{eq:WP-metric}.
\end{remark}
It is useful to introduce a basis for the rod vectors  where $\ell_0, \ell_1$ take the form $ (0,1)$ and $ (1,0)$, respectively. Expressing the normalized rod vectors $\ell_k$ in this basis yields a collection $\{v_0, \dots, v_n\}$ of 2-vectors with integer entries called the rod structure.  The integers 
\begin{align} \label{eq:pqDef}
p = \det (v_0 \ v_n), \quad q = \det (v_1 \ v_n),
\end{align} 
encode the topology of the asymptotic end.  We may without loss of generality assume $p \geq 0$, $q > 0$. In the ALE case,  the end has topology $\Reals_+ \times L(p,q)$, for coprime integers $p, q$, with $p > 0$, where $L(p,q)$ is a lens space. In the ALF case, $B = L(p,1)$, where we now allow $p=0$ in which case $B = S^1 \times S^2$.  
Thus $p = 0$ if and only if $(\cM, g_{ab})$ is AF.  See \cite{2021LMaPh.111..133K} for results on existence and uniqueness of AF Weyl-Papapetrou metrics with prescribed rod structure. 

\begin{remark} \label{rem:n3rodstructure}
In the case of $n=3$ turning points, the general regular, i.e. satisfying \eqref{eq:RegCond}, rod structure is given by
\begin{align}
v_0 = (0,1), \qquad
v_1 = (1,0), \qquad
v_2 = (-a,1), \qquad
v_3 = (1-ab,b),
\end{align}
with integers $a,b$, cf. \cite{Nilsson-Inst}. Assuming $p=0$ in \eqref{eq:pqDef} leads to $a  = b = \pm 1$ and without loss of generality one can choose the positive sign.
\end{remark}

\subsection{Null tetrad formalism}  \label{sec:nulltetrad}
In four dimensions complex geometry gives a unifying picture of Lorentzian and Riemannian geometry and the formalism we use here is guided by that fact. This perspective plays a fundamental role in the work of Penrose and collaborators, cf. e.g. \cite{PR:II}. To compute the connection and curvature we use a version of the Newman-Penrose null tetrad formalism, which was originally developed for Lorentzian signature to study General Relativity, but can be adapted to work for any signature and complex geometries, cf. \cite{2009arXiv0911.3364G} and references therein. The setup used here is such that the equations of \cite{PR:II} can be used with the appropriate interpretation. Let $(e_\ba^a)_{\ba = 1. \dots, 4}$ be a complex null tetrad, $(e_\ba^a) = (L^a, N^a, M^a, P^a)$, 
normalized by 
\begin{align} 
L^a N_a = 1, \qquad M^a P_a = -1, 
\end{align} 
all other inner products being zero. The metric takes the form
\begin{align} \label{eq:metric-tetrad} 
g_{ab} = 2  L_{(a} N_{b)} - 2 M_{(a} P_{b)}.
\end{align} 
For Riemannian signature, we additionally require
\begin{align} \label{RiemSignTetrad}
\bar L^a = N^a, \quad \bar M^a = -P^a.
\end{align} 
The volume form
\begin{align}
\epsilon_{abcd}={}&24 L_{[a}M_{c}N_{b}P_{d]},
\end{align}
is normalized such that $
\epsilon_{abcd} \epsilon^{abcd}= 24$. 
The dual of a 2-form $\omega_{ab}$ is defined by $ *\omega_{ab}= \frac{1}{2} \epsilon_{ab}{}^{cd}\omega_{cd}$. Recall that in a four-dimensional Riemannian space, ${*}{*}\omega_{ab} = \omega_{ab}$ so that 2-forms can be de\-composed into self dual and anti-self dual parts. 
A basis for the space of  2-forms can be expressed in terms of three self dual 2-forms,
\begin{align} \label{eq:ZDef}
Z^{0}_{ab}= 2 P_{[a}N_{b]}, \qquad
Z^{1}_{ab}= L_{[a}N_{b]}
 -  M_{[a}P_{b]}, \qquad
Z^{2}_{ab}= 2 L_{[a}M_{b]},
\end{align}
and three anti-self dual 2-forms,
\begin{align} \label{eq:ZtildeDef}
\tilde{Z}^{0}_{ab}= 2 M_{[a}N_{b]}, \qquad
\tilde{Z}^{1}_{ab}= L_{[a}N_{b]}
 +  M_{[a}P_{b]}, \qquad
\tilde{Z}^{2}_{ab}= 2 L_{[a}P_{b]}.
\end{align}
Here and below we use a tilde $\tilde{\hspace{1pt}}$ to indicate fields on the anti-self dual side. By \eqref{RiemSignTetrad} we have
\begin{align} 
\bar{Z}^{0}_{ab} = Z^{2}_{ab}, \quad \bar{Z}^{1}_{ab} = - Z^{1}_{ab}, \quad \bar{\tilde{Z}}^{0}_{ab} = \tilde{Z}^{2}_{ab}, \quad \bar{\tilde{Z}}^{1}_{ab} = - \tilde{Z}^{1}_{ab} .
\end{align} 
The connection also splits into self dual and anti-self dual parts. The tetrad components of the connection are encoded in a set of complex scalars called spin coefficients, denoted with greek letters. They are defined by
\begin{subequations}\label{eq:SpinCoeff}
\begin{align}
M^b \nabla_a L_b ={}& \tau L_a + \kappa N_a - \rho M_a -\sigma P_a, \\
\frac{1}{2}(N^b \nabla_a L_b - P^b \nabla_a M_b ) ={}& \gamma L_a + \epsilon N_a - \alpha M_a - \beta P_a, \\
N^b \nabla_a P_b ={}& \nu L_a + \pi N_a - \lambda M_a -\mu P_a.
\end{align}
\end{subequations} 
and
\begin{subequations} \label{eq:SpinCoefftilde}
\begin{align}
P^b \nabla_a L_b ={}& \tilde\tau L_a + \tilde\kappa N_a - \tilde\sigma M_a - \tilde\rho P_a, \\
\frac{1}{2}(N^b \nabla_a L_b + P^b \nabla_a M_b ) ={}& \tilde\gamma L_a + \tilde\epsilon N_a - \tilde\beta M_a - \tilde\alpha P_a, \\
N^b \nabla_a M_b ={}& \tilde\nu L_a + \tilde\pi N_a - \tilde\mu M_a - \tilde\lambda P_a.
\end{align}
\end{subequations}
From \eqref{RiemSignTetrad} it follows that
\begin{align} \label{RiemSignSpinCoeffIds}
\alpha ={}\bar{\beta},\quad
\epsilon ={}- \bar{\gamma},\quad
\kappa ={}\bar{\nu},\quad
\sigma ={}- \bar{\lambda},\quad
\rho ={}- \bar{\mu},\quad
\tau ={}\bar{\pi},
\end{align}
and the same for the tilded spin coefficients. This reduces the set of spin coefficients from $12 + 12$ to $6 + 6$ independent complex scalars. 

Introducing the Newman-Penrose Weyl curvature scalars
\begin{subequations} \label{Psidef}
\begin{align}
\Psi_{0} = W_{1331}, \quad
\Psi_{1} = W_{1231}, \quad
\Psi_{2} = W_{1324}, \quad
\Psi_{3} = W_{1224}, \quad
\Psi_{4} =  W_{2442},\\
\tilde{\Psi}_{0} =  W_{1441}, \quad
\tilde{\Psi}_{1} = W_{1241}, \quad
\tilde{\Psi}_{2} = W_{1423}, \quad
\tilde{\Psi}_{3} = W_{1223}, \quad
\tilde{\Psi}_{4} = W_{2332},
\end{align}
\end{subequations}
the self dual Weyl tensor takes the form 
\begin{align} \label{eq:W+Z}
W^+ ={}& -\Psi_0 Z^0 \odot Z^0 + 4\Psi_1 Z^0 \odot Z^1 - \Psi_2 (4 Z^1 \odot Z^1+2 Z^0 \odot Z^2) \nonumber \\ 
& \quad + 4\Psi_3 Z^1 \odot Z^2 - \Psi_4 Z^2 \odot Z^2, 
\end{align}  
where $\odot$ is the symmetrized tensor product, such that $(X\odot Y)_{abcd} = (X_{ab}Y_{cd}+Y_{ab}X_{cd})/2$. 
The anti-self dual Weyl tensor $W^-$ has the analogous form in terms of the tilded quantities.  Due to \eqref{RiemSignTetrad} we have 
\begin{align} \label{PsiRiemSign}
\bar\Psi_0 = \Psi_4, \qquad 
\bar\Psi_1 = -\Psi_3, \qquad 
\bar\Psi_2 = \Psi_2,
\end{align}
and the same for the tilded quantities, which reduces the set from $5+5$ complex to $2+2$ complex and $1+1$ real (a priori independent) scalars. 

The self dual Weyl tensor 
$W^+_{abcd}$ 
defines a self-adjoint endomorphism on the three-dimensional space of self dual 2-forms. In terms of the orthonormal basis of self dual 2-forms 
\begin{align} 
\left\{ \tfrac{i}{2}(Z^0-Z^2), \tfrac{1}{2}(Z^0+Z^2), i  Z^1\right\}, 
\end{align} 
it is represented by the symmetric matrix\footnote{The basis used here differs by a constant factor from the spinor basis used in \cite[p.234]{PR:II}.}  \cite[p.234]{PR:II}, 
\begin{align} \label{PsiMatrix}
\mathbf{\Psi} = 
\left(\begin{array}{ccc}
\Psi_{0}{} - 2 \Psi_{2}{} + \Psi_{4}{} & i \Psi_{0}{} - i \Psi_{4}{} & -2 \Psi_{1}{} + 2 \Psi_{3}{}\\
i \Psi_{0}{} - i \Psi_{4}{} & - \Psi_{0}{} - 2 \Psi_{2}{} -  \Psi_{4}{} & -2i \Psi_{1}{} - 2i \Psi_{3}{}\\
-2 \Psi_{1}{} + 2 \Psi_{3}{} & -2i \Psi_{1}{} - 2i \Psi_{3}{} & 4 \Psi_{2}{}.
\end{array}\right)
\end{align}
This form will be used to discuss the algebraic type of the Weyl tensor. The analogous statement holds for the anti-self dual Weyl tensor.

\subsection{Tetrad rotations} \label{sec:tetrad-rot} 
The six dimensional rotation group acts independently on self dual 2-forms \eqref{eq:ZDef} and anti-self dual 2-forms \eqref{eq:ZtildeDef}. Here we restrict to the self dual side. Decorating the transformed quantities with a hat $\hat{\hspace{5pt}}$, the general rotation with three real parameter functions $\chi, \theta, \psi$ of range  $[0, 2\pi)$, rotates the null tetrad via
\begin{subequations} \label{TetradRotation}
\begin{align}
\hat{L}_{a}={}& e^{i \theta} \cos\chi L_{a} + e^{i \psi}  \sin\chi P_{a} ,\\
\hat{N}_{a}={}& e^{-i \theta}\cos\chi N_{a} -  e^{-i \psi} \sin\chi M_{a},\\
\hat{M}_{a}={}& e^{i \theta} \cos\chi M_{a} + e^{i \psi}  \sin\chi N_{a} ,\\
\hat{P}_{a}={}& e^{-i \theta}\cos\chi P_{a} -   e^{-i \psi} \sin\chi L_{a}.
\end{align}
\end{subequations} 
The self dual 2-forms transform according to
\begin{subequations}
\begin{align}
\hat{Z}^{0}{}_{ab}={}&\frac{\cos^2\chi Z^{0}{}_{ab}}{e^{2i \theta}}
 -  \frac{\sin(2 \chi) Z^{1}{}_{ab}}{e^{i \psi + i \theta}}
 + \frac{\sin^2\chi Z^{2}{}_{ab}}{e^{2i \psi}}, \\
\hat{Z}^{1}{}_{ab}={}&\frac{e^{i \psi} \sin(2 \chi) Z^{0}{}_{ab}}{2 e^{i \theta}}
 + \cos(2 \chi) Z^{1}{}_{ab}
 -  \frac{e^{i \theta} \sin(2 \chi) Z^{2}{}_{ab}}{2 e^{i \psi}}, \\
\hat{Z}^{2}{}_{ab}={}&e^{2i \psi} \sin^2\chi Z^{0}{}_{ab}
 + e^{i (\psi + \theta)} \sin(2 \chi) Z^{1}{}_{ab}
 + e^{2i \theta} \cos^2\chi Z^{2}{}_{ab},
\end{align}
\end{subequations}
while the anti-self dual 2-forms are invariant. From \eqref{Psidef} it follows that the Weyl scalars transform according to
\begin{subequations} \label{PsiTrafo}
\begin{align}
\hat{\Psi}_{0}{}={}&e^{4i \theta} \cos^4\chi \Psi_{0}{}
 + 4 e^{i (\psi + 3 \theta)} \cos^3\chi \Psi_{1}{} \sin\chi
 + 6 e^{2i (\psi + \theta)} \cos^2\chi \Psi_{2}{} \sin^2\chi\nonumber\\
& - 4 e^{i (3 \psi + \theta)} \cos\chi \overline{\Psi}_{1}{} \sin^3\chi
 + e^{4i \psi} \overline{\Psi}_{0}{} \sin^4\chi ,\\
\hat{\Psi}_{1}{}={}&\tfrac{1}{2} e^{2i \theta} \cos(2 \chi) \Psi_{1}{}
 + \tfrac{1}{2} e^{ i  2 \theta} \cos(4 \chi) \Psi_{1}{}
 - \tfrac{1}{2} e^{2i \psi} \cos(2 \chi) \overline{\Psi}_{1}{}\nonumber\\
& +  \tfrac{1}{2} e^{ i 2 \psi } \cos(4 \chi) \overline{\Psi}_{1}{}
 -  \frac{e^{3i \theta} \Psi_{0}{} \sin(2 \chi)}{4 e^{i \psi}}\nonumber\\
& -  \frac{e^{3i \theta} \cos(2 \chi) \Psi_{0}{} \sin(2 \chi)}{4 e^{i \psi}}
 + \tfrac{3}{2} e^{i (\psi + \theta)} \cos(2 \chi) \Psi_{2}{} \sin(2 \chi)\nonumber\\
& + \frac{e^{3i \psi} \overline{\Psi}_{0}{} \sin(2 \chi)}{4 e^{i \theta}}
 -  \frac{e^{3i \psi} \cos(2 \chi) \overline{\Psi}_{0}{} \sin(2 \chi)}{4 e^{i \theta}},\\
\hat{\Psi}_{2}{}={}&\frac{e^{2i \theta} \Psi_{0}{}}{8 e^{2i \psi}}
 -  \frac{e^{2i \theta} \cos(4 \chi) \Psi_{0}{}}{8 e^{2i \psi}}
 + \tfrac{1}{4} \Psi_{2}{}
 + \tfrac{3}{4} \cos(4 \chi) \Psi_{2}{}
 + \frac{e^{2i \psi} \overline{\Psi}_{0}{}}{8 e^{2i \theta}}\nonumber\\
& -  \frac{e^{2i \psi} \cos(4 \chi) \overline{\Psi}_{0}{}}{8 e^{2i \theta}}
 -  \frac{e^{i \theta} \Psi_{1}{} \sin(4 \chi)}{2 e^{i \psi}}
 - \frac{e^{i \psi} \overline{\Psi}_{1}{} \sin(4 \chi)}{2 e^{i \theta}},\\
\hat{\tilde{\Psi}}_{i}{}={}&\tilde{\Psi}_{i}{}, \qquad i=0,1,2.
\end{align}
\end{subequations}

\subsection{Killing tetrad} 
Any torus symmetric geometry admits a tetrad of the following form.
\begin{definition} \label{def:KillingTetrad}
Let $(\cM, g_{ab})$ be a Riemannian four-manifold with coordinates $(\tau, \phi, x, y)$ such that $\partial_\tau, \partial_\phi$ are Killing fields. A complex null tetrad with 
\begin{align}
L^a = A(x,y) \partial_\tau^a + B(x,y) \partial_\phi^a,
\end{align}
is called a Killing tetrad. 
\end{definition}
If $(\cM, g_{ab})$  admits a Killing tetrad, this has many useful consequences. For example
\begin{subequations}
\begin{align} \label{VanishingSpinCoeffIsoTetrad}
\gamma ={}&0 = \epsilon,\quad
\sigma ={}0 = \lambda,\quad
\rho ={}0 = \mu,\\
\tilde{\gamma}={}&0 = \tilde{\epsilon},\quad
\tilde{\sigma}={}0 = \tilde{\lambda},\quad
\tilde{\rho}={}0 = \tilde{\mu},
\end{align}
\end{subequations}
which leaves $3 + 3$ independent spin coefficients. The first Cartan structure equation yields for a Killing tetrad
\begin{subequations} \label{eq:SpinCoeffExpandedForm}
\begin{align}
\alpha ={}&\tfrac{1}{4} N^{a} P^{b} \partial_{b}L_{a}
 -  \tfrac{1}{4} L^{a} P^{b} \partial_{b}N_{a}
 -  \tfrac{1}{2} M^{b} P^{a} \partial_{b}P_{a}
 + \tfrac{1}{2} M^{a} P^{b} \partial_{b}P_{a},\\
\kappa ={}&- L^{a} M^{b} \partial_{b}L_{a},\\
\pi ={}&\tfrac{1}{2} N^{a} P^{b} \partial_{b}L_{a}
 + \tfrac{1}{2} L^{a} P^{b} \partial_{b}N_{a}, \\
\tilde{\alpha}={}&\tfrac{1}{4} M^{b} N^{a} \partial_{b}L_{a}
 + \tfrac{1}{2} M^{b} P^{a} \partial_{b}M_{a}
 -  \tfrac{1}{2} M^{a} P^{b} \partial_{b}M_{a}
 -  \tfrac{1}{4} L^{a} M^{b} \partial_{b}N_{a},\\
\tilde{\kappa}={}&- L^{a} P^{b} \partial_{b}L_{a},\\
\tilde{\pi}={}&\tfrac{1}{2} M^{b} N^{a} \partial_{b}L_{a}
 + \tfrac{1}{2} L^{a} M^{b} \partial_{b}N_{a}.
\end{align}
\end{subequations}
We have
\begin{align}
\bar \pi = -\tilde \pi,
\end{align}
so that the connection is encoded in five complex spin coefficients.

The Weyl scalars in a Killing tetrad satisfy 
\begin{align}\label{PsiIsomTetrad}
\Psi_1 = 0 = \Psi_3, \qquad 
\tilde\Psi_1 = 0 = \tilde\Psi_3,
\end{align}
which leaves us with one complex and one real scalar on each side. Assuming that $(\cM, g_{ab})$ is Ricci flat, the second Cartan structure equation can be used to express the curvature scalars in terms of spin coefficients and their derivatives. Restricting to a Killing tetrad, we have 
\begin{subequations} \label{eq:PsiSimpForm}
\begin{align}
\Psi_{2}{}={}&\kappa \bar{\kappa}
 -  \pi \bar{\pi},\\
\tilde{\Psi}_{2}{}={}&\tilde{\kappa} \overline{\tilde{\kappa}}
 -  \pi \bar{\pi},\\
\Psi_{0}{}={}&
 -  (M^a \partial_a - 3 \bar{\alpha} -\tilde{\alpha}  -2 \bar{\pi})\kappa,\\
 \tilde{\Psi}_{0}{}={}&
 -  (P^a \partial_a  -3 \overline{\tilde{\alpha}}  - \alpha +2 \pi  )\tilde{\kappa}.
\end{align}
\end{subequations}
Define the curvature invariants 
\begin{align} 
I = W^+_{abcd} W^{+abcd}, \quad J = W^+_{ab}{}^{cd}W^+_{cd}{}^{ef}W^+_{ef}{}^{ab},
\end{align} 
and analogously $\tilde I, \tilde J$ on the anti-self dual side. Then $W^+_{abcd}$ is algebraically  special if and only if 
\begin{align} \label{eq:IJcond}
I^3 - 6 J^2 = 0 , 
\end{align}
cf. \cite[\S 8.3]{PR:II}. We have the following factorization. 
\begin{lemma} \label{rem:IJinvariants}
Let $(e_\ba^a)$ be a Killing tetrad. Then
\begin{align} \label{eq:IJcondKilling}
I^3 - 6 J^2 = 512 \Psi_{0}{} \Psi_{4}{} (\Psi_{0}{} \Psi_{4}{} -9 \Psi_{2}{}^2)^2,
\end{align}
and the Kretschmann scalar $\vert W \vert^2$ is given by 
\begin{align}\label{eq:Kretschmann}
 I + \tilde I = 24 \Psi_{2}{}^2
 + 8 \Psi_{0}{} \Psi_{4}{}
 + 24 \tilde{\Psi}_{2}{}^2
 + 8 \tilde{\Psi}_{0}{} \tilde{\Psi}_{4}{}.
\end{align}
\end{lemma}

\begin{remark}\label{rem:bfcalc}
Let $A = A(x,y)$. In a Killing tetrad we have 
\begin{align} \label{eq:dA}
(\d A)_a = -P_a M^b \partial_b A - M_a P^b \partial_b A.
\end{align}
Using the metric in Weyl-Papapetrou coordinates \eqref{eq:WP-metric}, contracting with $M^aP^b$ and using \eqref{eq:dA} for $\rho$ or $z$ leads to
\begin{align}
\mathbf{f} 
={}&- \frac{1}{2 (M^b \partial_b \rho) (P^c \partial_c \rho)}
=- \frac{1}{2 (M^b \partial_b  z) (P^c \partial_c z)}.
\end{align}
We also notice
\begin{align}
(M^b \partial_b z)^2 + (M^b \partial_b \rho)^2={}&0 
= (P^b \partial_b z)^2 + (P^b \partial_b \rho)^2.
\end{align}
\end{remark}

\section{The Chen-Teo metric} \label{sec:CTmetric} 

The Chen-Teo 5-parameter family \cite{2015PhRvD..91l4005C} of Ricci flat metrics is given in C-metric coordinates $(x,y,\tau,\phi)$ by 
\begin{align} \label{eq:ChenTeo5paramMetric}
\d s^2={}&\frac{k H}{ (x -  y)^3} \left( \frac{\d x^2}{X} -\frac{\d y^2}{Y} - \frac{XY}{kF}\d\phi^2  \right)  + \frac{(F \d\tau+ G \d\phi)^2}{F H (x - y)}.
\end{align}
The five real functions $H,F,G,X,Y$ depending only on $(x,y) \in \Reals^2$ are given by\footnote{In the remainder of this paper, the symbol $\nu$ will only refer to the parameter and not to a spin coeffcient, as it can be avoided by \eqref{RiemSignSpinCoeffIds}.} 
\begin{subequations} \label{CTMetricFunctions}
\begin{align}
H(x,y)={}&(\nu x + y) \bigl((\nu x -  y) (a_{1}{} -  a_{3}{} x y) - 2 (1 -  \nu) (a_{0}{} -  a_{4}{} x^2 y^2)\bigr), \label{CTMetricFunctions:H} \\
F(x,y)={}&X y^2
 -  x^2 Y,\\
G(x,y)={}&X (a_{0}{} \nu^2 + 2 a_{3}{} \nu y^3 -  a_{4}{} y^4 + 2 a_{4}{} \nu y^4)
 + (a_{0}{} - 2 a_{0}{} \nu - 2 a_{1}{} \nu x -  a_{4}{} \nu^2 x^4) Y,\\
X(x)={}&a_{0}{} + a_{1}{} x + a_{2}{} x^2 + a_{3}{} x^3 + a_{4}{} x^4,\\
Y(y) ={}&a_{0}{} + a_{1}{} y + a_{2}{} y^2 + a_{3}{} y^3 + a_{4}{} y^4,
\end{align}
\end{subequations}
with parameters 
\begin{align} \label{eq:CTparameters} 
(k, \nu,a_0, \dots, a_4) \in \Reals^7.
\end{align} 
It follows from \eqref{eq:ChenTeo5paramMetric} that $\partial_\tau$ and $\partial_\phi$ are Killing fields of $g_{ab}$ and the determinant is given by
\begin{align}
 \det g = \frac{k^2 H^2}{(x-y)^{10}}.
\end{align}

\begin{remark}
The Chen-Teo metric  \eqref{eq:ChenTeo5paramMetric} admits, in addition to the two Killing symmetries,  two continuous and two discrete symmetries, cf. \cite[\S II.A]{2015PhRvD..91l4005C}. Hence, two of the seven parameters $k, \nu, a_0, \dots, a_4$ are redundant, so that \eqref{eq:ChenTeo5paramMetric} represents a five-parameter family of line elements.  
\end{remark}
Let $x_1, \dots, x_4$ be the roots of $X$ such that\footnote{The additional assumption of four real roots is not necessary for Ricci flatness.} 
\begin{align} 
X(x) = a_4 (x - x_1)(x-x_2)(x-x_3)(x-x_4).
\end{align} 
The Chen-Teo metrics can alternatively be parametrized by 
$k, \nu, a_4, x_1, \dots, x_4$. 
The parameters $a_0, \dots, a_4$ and $a_4, x_1, \dots, x_4$  are related by
\begin{subequations} \label{eq:aToxEqs}
\begin{align}
a_{0}{}={}&a_{4}{} x_{1}{} x_{2}{} x_{3}{} x_{4}{},\\
a_{2}{}={}&a_{4}{} (x_{1}{} x_{2}{} + x_{1}{} x_{3}{} + x_{2}{} x_{3}{} + x_{1}{} x_{4}{} + x_{2}{} x_{4}{} + x_{3}{} x_{4}{}),\\
a_{1}{}={}&- a_{4}{} (x_{1}{} x_{2}{} x_{3}{} + x_{1}{} x_{2}{} x_{4}{} + x_{1}{} x_{3}{} x_{4}{} + x_{2}{} x_{3}{} x_{4}{}),\\
a_{3}{}={}&- a_{4}{} (x_{1}{} + x_{2}{} + x_{3}{} + x_{4}{}).
\end{align}
\end{subequations} 

\begin{lemma} \label{lem:FirstOrderIds} 
The metric functions \eqref{CTMetricFunctions} satisfy the following first order differential relations,
\begin{subequations} \label{eq:1stid}
\begin{align}
\partial_x \left( \frac{G y}{F Y}  + \frac{H x}{F (x -  y)} \right) ={}& - \frac{2 (1 -  \nu) (\nu x + y)}{(x -  y)^3},\label{eq:dxId}\\
\partial_y \left( \frac{G x}{F X} + \frac{H y}{F (x -  y)} \right) ={}& \frac{2 (1 -  \nu) (\nu x + y)}{(x -  y)^3},\label{eq:dyId}\\
(x \partial_x -  y \partial_y) H={}&- \frac{2 (1 -  \nu) F (\nu x + y)}{(x -  y)^2} + \frac{2 H (\nu x^2 + y^2)}{(x -  y) (\nu x + y)}. \label{eq:Hflow}
\end{align}
\end{subequations}
\end{lemma}
\begin{proof}
All three identities can be verified by a short computation. Additionally, we note the following algebraic identities among the metric functions,
\begin{align*}
\frac{G y}{F Y} + \frac{H x}{F (x -  y)} ={}&\frac{(1 -  \nu) \bigl(2 \nu x + (1 -  \nu) y\bigr)}{(x -  y)^2} -  \frac{a_{1}{} \nu^2 + a_{2}{} \nu^2 y - a_3(2 - \nu) \nu y^2 +a_4(1 - \nu)^2 y^3)}{Y}, \\
 \frac{G x}{F X} + \frac{H y}{F (x -  y)}={}&
 \frac{(1 - \nu) \bigl( 2 y - (1 - \nu) x \bigr)}{(x -  y)^2} + \frac{a_{1}{} + (1-2\nu)(a_{2}{}x+a_3 x^2)  + a_{4}{} (1 - \nu)^2 x^3}{X}.
\end{align*}
The second terms on the right hand sides are functions of only $y$ and $x$, respectively, from which \eqref{eq:dxId} and \eqref{eq:dyId} follow directly.
\end{proof}

\subsection{Rod structure} \label{sec:rodstruct}
To compute the rod structure via \eqref{eq:RodkGram}, we need the coefficient function $\mathbf{f}$  of the Chen-Teo  metric \eqref{eq:WP-metric} in Weyl-Papapetrou  coordinates  \cite[(3.12)]{2015PhRvD..91l4005C}
\begin{align}\label{rhozdef} 
\rho = \frac{\sqrt{-Y X }}{(x -  y)^2}, \qquad
z = \frac{(x + y) (a_{1}{} + a_{3}{} x y) + 2 (a_{0}{} + a_{2}{} x y + a_{4}{} x^2 y^2)}{2 (x -  y)^2}.
\end{align}

A short computation using Remark \ref{rem:bfcalc}, and the tetrad \eqref{eq:IsometryFrame} for the Chen-Teo metric below, yields\footnote{This corrects \cite[Eq. (2.21)]{2018JHEP...02..008B}.}
\begin{align}\label{bfeq1}
\mathbf{f} ={}&\frac{4 k H (x -  y)^3}{ X \bigl(4 Y + (x -  y) \partial_y Y\bigr)^2 - Y \bigl(4 X - (x -  y) \partial_x X\bigr)^2 }.
\end{align}

Following \cite[\S III]{2015PhRvD..91l4005C} assume $x_1 < x_2 < x_3$ and consider the domain
\begin{align} \label{eq:OmegaDef}
 \Domain = \{ (x,y) \in (x_2,x_3)\times (x_1,x_2) \} .
 \end{align}
The rods are located at the boundary
\begin{subequations} \label{eq:Rodk5param}
\begin{align}
\Rod_1 :{}&\quad  x=x_2, \quad x_1 < y <x_2 \\  
\Rod_2 :{}&\quad  y=x_1, \quad x_2 < x <x_3 \\  
\Rod_3 :{}&\quad  x=x_1, \quad x_1 < y <x_2 \\  
\Rod_4 :{}&\quad y=x_2, \quad x_2 < x <x_3 
\end{align} 
\end{subequations}
while the turning points are located at the corners of $\Domain$,
\begin{align} \label{eq:nuts}
\Nut_1={}& (x_2,x_1), \quad \Nut_2 = (x_3,x_1), \quad \Nut_3 =(x_3, x_2),
\end{align} 
with the fourth corner $(x_2,x_2)$ corresponding to the asymptotic end, see Figure~\ref{fig:domain}. In Weyl Papapetrou coordinates \eqref{rhozdef} the turning points are located at $\rho=0$ and
\begin{align}
z_1 = - \frac{a_4(x_1 x_2 + x_3 x_4)}{2}, \quad
z_2 = - \frac{a_4(x_1 x_3 + x_2 x_4)}{2}, \quad
z_3 = - \frac{a_4(x_2 x_3 + x_1 x_4)}{2}.
\end{align}
The rod structure of \eqref{eq:ChenTeo5paramMetric} can now be computed via \eqref{eq:RodkGram}, to be
\begin{subequations}\label{eq:ellk}
\begin{align}
\ell_0 =&
 \left(
\frac{2 \sqrt{k} \bigl( x_{1}{} x_{3}{} x_{4}{} - \nu^2 x_{2}{}^3 + 2 \nu x_{2}{} (x_{3}{} x_{4}{} + x_{1}{} x_{3}{} + x_{1}{} x_{4}{})\bigr)}{x_{12}\, x_{23}\, x_{24}} , \frac{2 \sqrt{k} x_{2}{}}{a_{4}{} x_{12}\, x_{23}\, x_{24}}
\right), \\
\ell_1 =& 
\left(
\frac{2 \sqrt{k} \bigl(x_{1}{}^3 - \nu^2 x_{2}{} x_{3}{} x_{4}{} + 2 \nu x_{1}{}^2 (x_{2}{} + x_{3}{} + x_{4}{})\bigr)}{x_{12}\, x_{13}\, x_{14}} , \frac{2 \sqrt{k} x_{1}{}}{a_{4}{} x_{12}\, x_{13}\, x_{14}}
\right), \\
\ell_2 =&
\left(
\frac{2 \sqrt{k} \bigl( x_{1}{} x_{2}{} x_{4}{} - \nu^2 x_{3}{}^3 + 2 \nu x_{3}{} (x_{2}{} x_{4}{} + x_{1}{}x_{2}{} + x_{1}{}x_{4}{})\bigr)}{x_{13}\, x_{23}\, x_{34}} , \frac{2 \sqrt{k} x_{3}{}}{a_{4}{} x_{13}\, x_{23}\, x_{34}}
\right), \\
\ell_3 =&
\left(
\frac{2 \sqrt{k} \bigl(x_{2}{}^3 - \nu^2 x_{1}{} x_{3}{} x_{4}{} + 2 \nu x_{2}{}^2 (x_{1}{} + x_{3}{} + x_{4}{})\bigr)}{x_{12}{}\, x_{23}{}\, x_{24}} , \frac{2 \sqrt{k} x_{2}{}}{a_{4}{} x_{12}{}\, x_{23}{}\, x_{24}}
\right),
\end{align}
\end{subequations}
where $x_{ij} = x_i - x_j$. The vectors $ v_i, i = 0, \dots, 3 $ of the rod structure are given by
\begin{align}\label{eq:vi}
v_i = A \ell_i, 
\quad
\text{with} 
\qquad
A = \frac{1}{\det(\ell_0, \ell_1)}
\left(\begin{array}{cc}
-\ell_0^\phi & \ell_0^\tau \\
\ell_1^\phi & -\ell_1^\tau
\end{array}\right).
\end{align}
Define the asymptotic nut charge as discussed in \cite[\S II.A]{2015PhRvD..91l4005C}, 
\begin{align}
\mathbf{n}={}&\frac{\sqrt{k} \bigl(2 \nu x_{2}{} (x_{1}{} x_{2}{} -  x_{1}{} x_{3}{} + x_{2}{} x_{3}{} -  x_{1}{} x_{4}{} + x_{2}{} x_{4}{} -  x_{3}{} x_{4}{}) -  (1 + \nu^2) ( x_{1}{} x_{3}{} x_{4}{}- x_{2}{}^3)\bigr)}{2 \sqrt{1 -  \nu^2} (x_{2}{} - x_{1}{}) (x_{2}{} -  x_{3}{}) (x_{2}{} -  x_{4}{})}.
\end{align}
If the metric is regular, then in view of \eqref{eq:pqDef} and Remark \eqref{rem:n3rodstructure} the components
\begin{align} \label{eq:pqCT}
p ={}&  - v_3[1]  = -\frac{2 \sqrt{1 -  \nu^2} \mathbf{n} (x_{1}{} -  x_{2}{}) x_{2}{} (x_{1}{} -  x_{3}{}) (x_{1}{} -  x_{4}{})}{\sqrt{k} (x_{1}{} + \nu x_{2}{})^2 (x_{1}{} x_{2}{} -  x_{3}{} x_{4}{})}, \\
q ={}& v_3[2] = \frac{(x_{1}{} -  x_{2}{}) \Bigl(x_{1}{}^2 x_{2}{} + \nu^2 x_{2}{} x_{3}{} x_{4}{} + x_{1}{} \bigl(x_{2}{}^2 + \nu^2 x_{3}{} x_{4}{} + 2 \nu x_{2}{} (x_{3}{} + x_{4}{})\bigr)\Bigr)}{(x_{1}{} + \nu x_{2}{})^2 (x_{1}{} x_{2}{} -  x_{3}{} x_{4}{})},
\end{align}
must be integers. We emphasize that, in contrast to $\mathbf{n}$, the quantities $p$ and $q$ do not depend on the parameter $k$. The AF case, which has $p = 0, q= 1$, is discussed in Section \ref{sec:CTInst}.

\begin{remark}
Note that the for the function $\mathbf{f}$ defined in \eqref{bfeq1}, $\sqrt{\rho^2+(z-z_i)^2}\mathbf{f}$ with $z_i = z(n_i)$ has a limit at the turning points. This is compatible with the discussion in \cite[Appendix C]{BAIRD2021104310}.
\end{remark}

\subsection{Special geometry} \label{sec:AlgSpec}

A Killing tetrad for the Chen-Teo metric \eqref{eq:ChenTeo5paramMetric} is given by
\begin{subequations} \label{eq:IsometryFrame}
\begin{align}
L^{a}={}&\frac{1}{\sqrt{2}}\left(\sqrt{\frac{H (x -  y)}{F}}\partial_\tau^{a}  + i  \sqrt{\frac{F (x -  y)^3}{H X (-Y)}}\left(\partial_\phi^{a} -  \frac{G}{F}\partial_\tau^{a} \right)\right),\\
N^{a}={}&\frac{1}{\sqrt{2}}\left(\sqrt{\frac{H (x -  y)}{F}}\partial_\tau^{a}  - i  \sqrt{\frac{F (x -  y)^3}{H X (-Y)}}\left(\partial_\phi^{a} -  \frac{G}{F}\partial_\tau^{a} \right)\right),\\
M^{a}={}&\frac{i\sqrt{(x -  y)^3}}{\sqrt{2k H}}\left( \sqrt{X}\partial_x^{a} + i \sqrt{-Y}\partial_y^{a}\right),\\
P^{a}={}&\frac{i\sqrt{(x -  y)^3}}{\sqrt{2k H}}\left( \sqrt{X}\partial_x^{a} - i \sqrt{-Y}\partial_y^{a}\right).
\end{align}
\end{subequations}
Then the tetrad $L^a, N^a, M^a, P^a$ satisfies \eqref{RiemSignTetrad} and $g_{ab}$ can be expressed via \eqref{eq:metric-tetrad}.
We have
\begin{align} \label{eq:orientation}
L \wedge N \wedge M \wedge P = \sqrt{\det g}\, d x \wedge d y \wedge d\tau \wedge d\phi.
\end{align}

To express the spin coefficients in coordinates, it is convenient to  define the complex function
\begin{align} \label{eq:fDef}
f={}&x \sqrt{-Y} + i y\sqrt{X},
\end{align}
so that
\begin{align}
F = f \bar f.
\end{align}

\begin{prop} \label{lem:SpinCoeff5ParamFamily}
The non-vanishing spin coefficients $\pi, \kappa, \alpha, \tilde\pi, \tilde\kappa, \tilde\alpha$ in the tetrad \eqref{eq:IsometryFrame}, taking \eqref{RiemSignSpinCoeffIds} into account, are given by
\begin{subequations} \label{eq:sdspincoeff5param}
\begin{align}
 \frac{2 \sqrt{2k H} }{(x -  y)^{3/2}}  \pi 
={}& (i \sqrt{X} \partial_x  + \sqrt{-Y} \partial_y ) \log{\frac{\sqrt{-YX}}{(x-y)^2}}, \\
\frac{4  \bar{f} \sqrt{2 k H} }{f (x -  y)^{3/2}} \kappa
={}&- \frac{4 \bigl(i \nu \sqrt{X} + \sqrt{- Y}\bigr)}{\nu x + y}
 + \frac{i \partial_x X}{\sqrt{X}}
 -  \frac{ \partial_y Y}{\sqrt{- Y}}, \\
\frac{8  F \sqrt{2k H} }{(x -  y)^{3/2}} \alpha
={}&4i \bar{f} \bigl(\frac{i (1 + \nu) F}{(x -  y) (\nu x + y)} + \sqrt{-YX} \bigr)
 -  (f + \bar{f}) y \partial_x X + (f -  \bar{f}) x \partial_y Y,
\end{align}
\end{subequations}
and 
\begin{subequations}
\begin{align}
\tilde{\pi}={}&- \bar{\pi}, \\
\frac{4  \bar{f} \sqrt{2 k H} }{f (x -  y)^{3/2}} \tilde{\kappa}
={}&  \frac{4 (1 - \nu) \bar{f}^3 (\nu x + y)}{H xy (x -  y)^2 }
- \frac{4 \bar{f}}{x y}
 + \frac{2 \bar{f} \partial_x H}{H y}
 + \frac{2 \bar{f} \partial_y H}{H x}
  + \frac{i \partial_x X}{\sqrt{X}}
 + \frac{ \partial_y Y}{\sqrt{- Y}}, \\
\frac{8  F \sqrt{2 k H} }{(x -  y)^{3/2}} \tilde{\alpha}
={}& \frac{4 (1 - \nu) F^2 f (\nu x + y)}{H x (x -  y)^2 y}
 + \frac{4 \bigl(F f x + F \bar{f} (- x + y)  + f x^2 (x -  y) Y\bigr)}{x (x -  y) y}\nonumber\\
& + \frac{2 F \bar{f} \partial_x H}{H y}
 + \frac{2 F \bar{f} \partial_y H}{H x}
   -  (f + \bar{f}) y \partial_x X
 - ( f - \bar{f}) x \partial_y Y.
\end{align}
\end{subequations}
\end{prop}
\begin{proof}
Use the spin coefficients in the form \eqref{eq:SpinCoeffExpandedForm} and Lemma \ref{lem:FirstOrderIds} to eliminate derivatives of $G$ and $H$.
\end{proof}
Note that there are no derivatives of $H$ in \eqref{eq:sdspincoeff5param}.
\begin{prop}
The non-vanishing self dual curvature components in the tetrad \eqref{eq:IsometryFrame}, taking \eqref{PsiRiemSign} into account, are given by
\begin{align}\label{eq:PsiValues}
\Psi_{2}{}=\frac{(1 + \nu) (x -  y)^3}{4 k (\nu x + y)^3},
&&
\Psi_{0}{}=\frac{3 f \Psi_{2}{}}{\bar{f}},
\end{align}
where $f$ is defined in \eqref{eq:fDef}. 
\end{prop}
\begin{proof}
Using Ricci identities in the form \eqref{eq:PsiSimpForm} and Proposition \ref{lem:SpinCoeff5ParamFamily} together with Lemma \ref{lem:FirstOrderIds} yields the result.
\end{proof}
Now we can state the following theorem, which is also the main ingredient in Theorem \ref{thm:main-intro}.
\begin{thm}  \label{thm:AlgNonSpec}
The Chen-Teo family of four-dimensional Ricci-flat Riemannian metrics \eqref{eq:ChenTeo5paramMetric} in the orientation given by \eqref{eq:orientation} has Weyl tensor of  type $O^+ I^-$ for $\nu=-1$, $D^+ D^-$ for $\nu=1$ and $D^+ I^-$ for $-1< \nu < 1$.
\end{thm}
\begin{proof}
The eigenvalues of \eqref{PsiMatrix} in a Killing tetrad are given by
\begin{align} \label{PSIroots}
\{4 \Psi_2,
-2(\Psi_2 - \sqrt{\Psi_0}\sqrt{\Psi_4}),
-2(\Psi_2 + \sqrt{\Psi_0}\sqrt{\Psi_4})\}.
\end{align}
From \eqref{eq:PsiValues} we notice that
\begin{align} \label{eq:typeDDegEigenvalueCondition}
\Psi_0 \overline\Psi_0  =  9 \Psi_2^2,
\end{align}
and hence due to \eqref{PsiRiemSign} this implies that $W^+_{abcd}$ is algebraically special by Remark \ref{rem:IJinvariants}. The eigenvalues simplify to 
\begin{align}
\{4\Psi_2, 4\Psi_2, -8 \Psi_2\}.
\end{align}
From \eqref{eq:PsiValues} we find $W^+_{abcd}=0$ for $\nu=-1$, which corresponds to the half-flat Gibbons-Hawking family of instantons.

The anti-self dual curvature components $\tilde{\Psi}_{0}{}, \tilde{\Psi}_{2}{}$ can be evaluated from \eqref{eq:PsiSimpForm} and Lemma~\ref{lem:SpinCoeff5ParamFamily}, eg.
\begin{align} \label{eq:Psi2tilde}
\frac{4 k H^3 \tilde{\Psi}_{2}{}}{(x -  y)^3 (\nu x + y)^3}={}&
 a_{1}{}^2 \Bigl(3 a_{3}{} (1 - \nu) x y + a_{4}{} \bigl(2 \nu x^3 + 6 (1 - 2 \nu) x^2 y + 6 (2 - \nu) x y^2 - 2 y^3\bigr)\Bigr)\nonumber\\
& +  a_{1}{} (1 - \nu) \Bigl(4 a_{0}{} (a_{2}{} + 6 a_{4}{} x y) + x^2 y^2 \bigl(-3 a_{3}{}^2 + 4 a_{4}{} (3 a_{2}{} + 2 a_{4}{} x y)\bigr)\Bigr) \nonumber\\
& - 4 a_{2}{} a_{3}{} (1 - \nu) x y (3 a_{0}{} + a_{4}{} x^2 y^2) - 8 a_{0}{} a_{3}{} (1 - \nu) (a_{0}{} + 3 a_{4}{} x^2 y^2)\nonumber\\
& + 2 a_{0}{} a_{3}{}^2 \bigl(- \nu x^3 - 3 (1 - 2 \nu) x^2 y + 3 (-2 + \nu) x y^2 + y^3\bigr)\nonumber\\
&-a_{1}{}^3 (1 - \nu) +  a_{3}{}^3 (1 - \nu) x^3 y^3.
\end{align}
An expression for $\tilde{\Psi}_{0}{}$ is too long to be displayed here, but it can easily be handled with computer algebra. For $\nu=1$ the expressions simplify to
\begin{align}
\tilde{\Psi}_{2}{}={}&- \frac{(a_{0}{} a_{3}{}^2 -  a_{1}{}^2 a_{4}{}) (x -  y)^3}{2 k (a_{1}{} -  a_{3}{} x y)^3}, 
&&
\tilde{\Psi}_{0}{}=\frac{3 f \tilde{\Psi}_{2}{}}{\bar{f}},
\end{align}
showing that $W^{-}_{abcd}\big{|}_{\nu=1}$ is algebraically special. This is the Plebanski-Demianski family of type $D^+ D^-$. 

For the remaining case it is sufficient to show that $W^{-}_{abcd}$ is not algebraically special at one point. We chose the center of the domain $\Domain$ at $x =(x_2+x_3)/2, y = (x_1+x_2)/2 $. Evaluating the non-trivial factor on the right hand side of the anti-self dual version of \eqref{eq:IJcondKilling} leads to
\begin{align} \label{eq:IJtilde}
\tilde{\Psi}_{0}{} \overline{\tilde{\Psi}}_{0}{} - 9 \tilde{\Psi}_{2}^2={}&36 B^{-1} (1 - \nu) (x_{1}{} -  x_{2}{})^2 (x_{1}{} + x_{2}{} - 2 x_{3}{}) (x_{1}{} -  x_{3}{})^6 (2 x_{1}{} -  x_{2}{} -  x_{3}{}) (x_{2}{} -  x_{3}{})^2 \times \nonumber \\
& (x_{1}{} + x_{2}{} - 2 x_{4}{}) (x_{2}{} + x_{3}{} - 2 x_{4}{}) ( x_{1}{} x_{4}{} - x_{2}{} x_{3}{}) (x_{1}{} x_{3}{} -  x_{2}{} x_{4}{}) (x_{1}{} x_{2}{} -  x_{3}{} x_{4}{}),
\end{align}
with denominator
\begin{align}
B={}&k^2 \biggl(x_{1}{}^3 (x_{2}{} + x_{3}{}) -  \nu x_{2}{} x_{3}{} (x_{2}{} + x_{3}{})^2 -  x_{2}{} (x_{2}{} - 3 x_{3}{}) \bigl(( \nu -1) x_{2}{} + \nu x_{3}{}\bigr) x_{4}{} \nonumber\\
&+ x_{1}{}^2 \bigl(2 x_{2}{} (x_{2}{} -  x_{3}{}) - 3 (x_{2}{} + x_{3}{}) x_{4}{}\bigr) + x_{1}{} \Bigl(x_{2}{}^3 -  \nu x_{3}{}^2 (x_{3}{} - 3 x_{4}{}) \nonumber \\
&+ 2 x_{2}{} x_{3}{} (\nu x_{3}{} + 5 x_{4}{} - 5 \nu x_{4}{}) + x_{2}{}^2 \bigl(3 ( \nu -1) x_{3}{} + ( 3 \nu -2) x_{4}{}\bigr)\Bigr)\biggr)^5.
\end{align}
Assuming the $x_i$ to be ordered and $\nu \neq 1$, there are only three exceptional values $x_4 \in \{\frac{x_2 x_3}{x_1},\frac{x_1 x_3}{x_2},\frac{x_1 x_2}{x_3}\}$ for which \eqref{eq:IJtilde} vanishes. However, one can check that the function $H$ vanishes at a turning point for each of the values, which would correspond to a curvature singularity.
\end{proof}
We note that for $\nu\neq -1$, both $W^+_{abcd} $ and  $W^-_{abcd} $ are non-zero, so the metric is non-Kähler.

\begin{remark}\label{rem:AlgSpec} 
For coordinates $(x,y)$ such that at least one of the Killing vectors vanishes, i.e. $XY=0$ and in particular on the rods $\Rod_i, i=1, \dots, 4$,  the Weyl tensor is of  type $D^+ D^-$ for $-1< \nu \leq 1$ and $O^+ D^-$ for $\nu=-1$. We find
\begin{align}\label{eq:Psi0tildeRod}
\tilde{\Psi}_{0}{}= 
\begin{cases}
3 \tilde{\Psi}_{2}{} &\text{ if } X = 0, \\
-3 \tilde{\Psi}_{2}{} &\text{ if } Y = 0, 
\end{cases}
\end{align}
which implies that that $W^{-}_{abcd}\big{|}_{XY=0}$ is algebraically special.
\end{remark}

For a geometry of Petrov type $D^+$, there exist adapted tetrads $\{\hat{L}^a, \hat{N}^a, \hat{M}^a, \hat{P}^a\}$ , such that the transformed curvature satisfies $\hat{\Psi}_{0}{}= 0, \hat{\Psi}_{1}{}= 0$.
\begin{lemma}
Starting from the tetrad \eqref{eq:IsometryFrame}, an adapated tetrad is given by
\begin{subequations} \label{eq:AdaptedTetrad}
\begin{align}
\hat{L}_{a}={}& \frac{1}{\sqrt{2}} \left( L_{a} + e^{i \psi} P_{a} \right),\\
\hat{N}_{a}={}& \frac{1}{\sqrt{2}} \left( N_{a} - e^{- i \psi} M_{a} \right),\\
\hat{M}_{a}={}& \frac{1}{\sqrt{2}} \left( M_{a} + e^{i \psi} N_{a} \right),\\
\hat{P}_{a}={}& \frac{1}{\sqrt{2}} \left( P_{a} -  e^{-i \psi} L_{a} \right),
\end{align}
\end{subequations} 
which is \eqref{TetradRotation} with parameters $\chi = \pi/4, \theta = 0$ and
\begin{align} \label{eq:psiTrafo}
e^{i \psi}=\left(\frac{\Psi_{0}{}}{\bar{\Psi}_{0}{}}\right)^{1/4} = i \sqrt{\frac{f}{\bar{f}}}.
\end{align}
The self dual curvature components take the form
\begin{align}\label{eq:Psihat}
\hat{\Psi}_{0}{}= 0, &&
\hat{\Psi}_{1}{}= 0, &&
\hat{\Psi}_{2}{}= - 2 \Psi_{2}{}.
\end{align}
\end{lemma}
\begin{proof}
Starting from a Killing tetrad, the relevant curvature components of \eqref{PsiTrafo} simplify to
\begin{subequations} 
\begin{align}
\hat{\Psi}_{0}{}={}&e^{4i \theta} \cos^4\chi \Psi_{0}{}
 + 6 e^{2i (\psi + \theta)} \cos^2\chi \Psi_{2}{} \sin^2\chi
 + e^{4i \psi} \overline{\Psi}_{0}{} \sin^4\chi ,\\
\hat{\Psi}_{1}{}={}& -  \frac{e^{3i \theta} \Psi_{0}{} \sin(2 \chi)}{4 e^{i \psi}}
 -  \frac{e^{3i \theta} \cos(2 \chi) \Psi_{0}{} \sin(2 \chi)}{4 e^{i \psi}}
 + \tfrac{3}{2} e^{i (\psi + \theta)} \cos(2 \chi) \Psi_{2}{} \sin(2 \chi)\nonumber\\
& + \frac{e^{3i \psi} \overline{\Psi}_{0}{} \sin(2 \chi)}{4 e^{i \theta}}
 -  \frac{e^{3i \psi} \cos(2 \chi) \overline{\Psi}_{0}{} \sin(2 \chi)}{4 e^{i \theta}},\\
\hat{\Psi}_{2}{}={}&\frac{e^{2i \theta} \Psi_{0}{}}{8 e^{2i \psi}}
 -  \frac{e^{2i \theta} \cos(4 \chi) \Psi_{0}{}}{8 e^{2i \psi}}
 + \tfrac{1}{4} \Psi_{2}{}
 + \tfrac{3}{4} \cos(4 \chi) \Psi_{2}{}
 + \frac{e^{2i \psi} \overline{\Psi}_{0}{}}{8 e^{2i \theta}} -  \frac{e^{2i \psi} \cos(4 \chi) \overline{\Psi}_{0}{}}{8 e^{2i \theta}}.
\end{align}
\end{subequations} 
We want the first two components to vanish and this can be achieved by substituting $\Psi_2$ using \eqref{eq:typeDDegEigenvalueCondition} and setting $\chi = \pi/4, \theta = 0$. The equations simplify to
\begin{subequations} 
\begin{align}
\hat{\Psi}_{0}{}={}&\tfrac{1}{4} \Psi_{0}{}
 + \tfrac{1}{4} e^{4i \psi} \bar{\Psi}_{0}{}
 -  \tfrac{1}{2} e^{2i \psi} (\Psi_{0}{} \bar{\Psi}_{0}{})^{1/2},\\
\hat{\Psi}_{1}{}={}&- \frac{\Psi_{0}{} -  e^{4i \psi} \bar{\Psi}_{0}{}}{4 e^{i \psi}},\\
\hat{\Psi}_{2}{}={}&\frac{2 \Psi_{0}{} + 2 e^{4i \psi} \bar{\Psi}_{0}{} + \tfrac{4}{3} e^{2i \psi} (\Psi_{0}{} \bar{\Psi}_{0}{})^{1/2}}{8 e^{2i \psi}}.
\end{align}
\end{subequations} 
and yield \eqref{eq:Psihat} with the choice \eqref{eq:psiTrafo}.
\end{proof}

Using \eqref{eq:W+Z}, the self dual Weyl curvature is now of the form 
\begin{align} 
W^+ = -\hat{\Psi}_2 (4 \hat Z^1 \odot \hat Z^1 + 2 \hat Z^0 \odot \hat Z^2),
\end{align} 
so that with $|W^+|^2 = W^+_{abcd} W^+{}^{abcd}$, we have 
\begin{align} \label{eq:W+norm}
|W^+|^2 = 24 \hat{\Psi}_2^2 = 96 \Psi_2^2.
\end{align} 

\begin{remark}
The middle self dual 2-form \eqref{eq:ZDef} in the adapted tetrad \eqref{eq:AdaptedTetrad} takes the simple form
\begin{align} \label{eq:Z1hat}
\hat{Z}^{1}_{ab}={}&\alpha_{[a}\d\tau_{b]}
 + \beta_{[a}\d\phi_{b]}.
\end{align}
where
\begin{align}
\alpha_{a}={}& \frac{i \sqrt{k} (x \d y_{a} - y \d x_{a})}{(x -  y)^2},\\
\beta_{a}={}& \frac{i \sqrt{k} \Bigl(  \bigl(G x (x -  y) + H X y\bigr)\d y_{a} - \bigl(G y (x -  y)  + H Y x\bigr)\d x_{a} \Bigr) }{F (x -  y)^3}.
\end{align}
\end{remark}

\begin{remark}
By the Goldberg-Sachs theorem \cite{2009arXiv0911.3364G}, algebraic degeneracy of the Weyl tensor is equivalent to vanishing of the spin coefficients $\kappa$ and $\sigma$ in an adapted tetrad. They are given by
\begin{align}
\kappa = M^a L^b \nabla_b L_a, \qquad 
\sigma = M^a M^b \nabla_b L_a.
\end{align}
Performing a tetrad rotation \eqref{TetradRotation} with 
\begin{align}
\chi=\pi/4, \qquad
\theta = 0, \qquad
  \partial_\tau \psi = 0 =  \partial_\phi \psi, 
\end{align}
 and using \eqref{VanishingSpinCoeffIsoTetrad} leads to
\begin{align}
\hat\kappa 
=& \hat M^a \hat L^b \nabla_b \hat L_a 
=\frac{\kappa}{2 \sqrt{2}} -\frac{1}{4 \sqrt{2}}  \bigl(
P^a \partial_a - 4 \alpha - 2\pi \bigr) e^{2i \psi}, \\ 
\hat\sigma 
=& \hat M^a \hat M^b \nabla_b \hat L_a 
=\frac{e^{3i \psi} \nu }{2 \sqrt{2}}+
\frac{e^{3i \psi} }{4 \sqrt{2}}
\bigl(M^a \partial_a + 4 \beta + 2 \tau \bigr) e^{-2 i \psi}.
\end{align}

Define a function $K$ by
\begin{align}
 K = Y \partial_y \left( \frac{Gx}{F} + \frac{HXy}{F(x-y)} \right) + X \partial_x \left( \frac{Gy}{F} + \frac{HYx}{F(x-y)} \right).
\end{align}
We do not have a geometric interpretation of this function, but by Lemma \ref{lem:FirstOrderIds}, we find 
\begin{align} \label{Keq0}
 K = 0.
\end{align}
The choice \eqref{eq:psiTrafo} for $\psi$ leads to
\begin{align}
\hat\kappa = \frac{f K (x-y)^{5/2}}{4i \sqrt{k X (-Y)H^3}}, \qquad
\hat\sigma =  \sqrt{\frac{f}{\bar f}}\frac{f K (x-y)^{5/2}}{4 \sqrt{k X (-Y)H^3}},
\end{align}
which for \eqref{Keq0} shows the algebraic degeneracy.
\end{remark}

Using the results of this section together with those of Penrose and Walker \cite{1970CMaPh..18..265W} and Derdzinski \cite{MR707181}, we can conclude the following facts
\begin{cor} \label{cor:details}
Let  $g_{ab}$ be given by \eqref{eq:ChenTeo5paramMetric} and define 
\begin{align} \label{eq:udef} 
u ={}& - \frac{x -  y}{ \nu x + y}.
\end{align} 
The following statements are locally valid. 

\begin{enumerate}
\item If $\nu  \ne -1$, the function $u$ given by \eqref{eq:udef} satisfies 
\begin{align} \label{eq:uTowP}
u = c |W^+|^{1/3}, \quad \text{with } 
c = \left(\frac{k}{\sqrt{6}(1+\nu)}\right)^{1/3}
\end{align} 
In particular for $\nu = -1$, $W^+_{abcd} = 0$ and $u=1$. 
\item An integrable complex structure is given by
\begin{align}
 J_a{}^b = 2 i g^{bc} \hat{Z}^1_{ac},
\end{align}
with  $\hat{Z}^1_{ac}$ given by \eqref{eq:Z1hat}.
\item The conformally transformed metric
\begin{align}
\bar g_{ab} = u^2 g_{ab} 
\end{align} 
is K\"ahler, with K\"ahler form 
\begin{align} 
\bar \omega_{ab} = J_{a}{}^c \, \bar g_{cb}.
\end{align} 
\item $Y_{ab} =  -\frac{2i}{3} u^{-1} \hat Z^1_{ab}$ is conformal Killing-Yano and
\begin{align}
\xi^a = g^{ac}\nabla^bY_{cb} = \bar{g}^{ac}J_{c}{}^b \nabla_b u = \frac{1+\nu}{\sqrt{k}} (\partial_\tau)^a .
\end{align}  
In particular, the Killing field $\xi^a$ is Hamiltonian and $\bar g_{ab}$ is extremal K\"ahler.
\item 
For $\nu \neq -1$ the scalar curvature of $\bar g_{ab}$ is 
\begin{align} 
\text{scal}_{\bar g} = u,
\end{align} 
while for $\nu = -1$, $\bar g = g$ and $\text{scal}_{\bar g} = 0$.  
\end{enumerate} 
\end{cor} 

\begin{remark}
We note that the relation \eqref{eq:uTowP} is such that the construction extends to the $W^+_{abcd} = 0$ case. 
The analogous choice for the Kerr solution would be $u = r- a \cos\theta$ which is also non-trivial in the mass $ M \to 0$ limit.
\end{remark}

\section{The Chen-Teo Instanton} \label{sec:CTInst} 

We now restrict the 5-parameter family of metrics to the 2-parameter family of Chen-Teo instantons and discuss its regularity and asymptotics. Although these properties have been discussed in previous works, they are sufficiently important to warrant a direct and independent treatment, which we provide here. Section~\ref{sec:ProofMainThm} collects the results to prove the main theorem.

Imposing the AF condition, ie. $p = 0, q = 1$, the rod structure discussed in Remark~\ref{rem:n3rodstructure} simplifies to
\begin{align}
v_0 = (0,1), \qquad
v_1 = (1,0), \qquad
v_2 = (-1,1), \qquad
v_3 = (0,1).
\end{align}
Equating this to the rod structure \eqref{eq:vi} of the Chen-Teo metric leads to the following set of four conditions on the parameters,
\begin{subequations}\label{eq:AFRegCond}
\begin{align}
-1={}&\frac{(x_{1}{} -  x_{2}{}) (x_{1}{} -  x_{4}{}) \bigl(\nu^2 x_{2}{} x_{3}{} (x_{2}{} + x_{3}{}) + x_{1}{} (x_{2}{} + x_{3}{}) x_{4}{} + 2 \nu x_{2}{} x_{3}{} (x_{1}{} + x_{4}{})\bigr)}{(x_{1}{} + \nu x_{2}{})^2 (x_{3}{} -  x_{4}{}) (x_{1}{} x_{2}{} -  x_{3}{} x_{4}{})},\\
1={}&\frac{(x_{1}{} -  x_{2}{}) (x_{1}{} + \nu x_{3}{})^2 (x_{2}{} -  x_{4}{}) (x_{1}{} x_{3}{} -  x_{2}{} x_{4}{})}{(x_{1}{} + \nu x_{2}{})^2 (x_{1}{} -  x_{3}{}) (x_{3}{} -  x_{4}{}) (x_{1}{} x_{2}{} -  x_{3}{} x_{4}{})},\\
0={}&p,\\
1={}&q,
\end{align}
\end{subequations}
with $p,q$ given by \eqref{eq:pqCT}.
\begin{lemma}
The conditions \eqref{eq:AFRegCond} imply the relations
\begin{align}
x_{1}{}= - \frac{2 \nu^2 x_{2} -  \sqrt{- 2\nu (1 + \nu)^2 x_{2}{}^2}}{2 + 4 \nu},\qquad
x_{3}{}= - \frac{2 x_{2} +  \sqrt{- 2\nu (1 + \nu)^2 x_{2}{}^2}}{4 \nu + 2 \nu^2},\qquad
x_{4}{}= 0.
\end{align}
This can be parametrized by $\xi$ in the form
\begin{align}\label{eq:nuxxi} 
\nu = -2\xi^2, \quad x_1 = - \frac{\xi(1-2\xi+2\xi^2)x_2}{1-2\xi}, \quad x_3 = \frac{(1-2\xi+2\xi^2)x_2}{4\xi^2(1-\xi)}, \quad x_4 = 0.
\end{align} 
For $x_2 < 0$ and
\begin{align}\label{eq:xicond}  
\xi \in (1/2,1/\sqrt{2})
\end{align} 
we have
\begin{align} \label{eq:xiordering}
 x_1 < x_2 < x_3< 0 = x_4.
\end{align}
\end{lemma}
\begin{definition}  \label{def:C-T-inst}
The Chen-Teo AF gravitational instanton with parameters $k>0, \xi \in (1/2,1/\sqrt{2})$ is given by \eqref{eq:ChenTeo5paramMetric}  with $a_4 > 0$, $x_2 < 0$ and with $x_1, x_3, x_4$ given by \eqref{eq:nuxxi}. The coordinates $(x,y)$ take values in $\bar\Omega\setminus\{(x_2,x_2)\}$ with $\Omega$ given in \eqref{eq:OmegaDef}. The range for the Killing coordinates is determined, as described in Section \ref{sec:rods}, by the rod structure \eqref{eq:vi} corresponding to the normalized rods $\ell_k$ given by \eqref{eq:ellk}. 
\end{definition}

\begin{remark} 
\begin{enumerate}
\item Taking into account the symmetries of \eqref{eq:ChenTeo5paramMetric}, there are other sets of parameters representing the same instanton, cf. \cite[\S IV.C]{2015PhRvD..91l4005C}. 
\item The Chen-Teo metric depends only on the ratios of the roots $\{x_1, \dots, x_4\}$, and hence $x_2<0$ can, without loss of generality, be fixed to any value.  
\end{enumerate}
\end{remark}

\subsection{Regularity} \label{sec:regularity} 
We shall now discuss the regularity of the Chen-Teo instanton. In order for $g_{ab}$ given by \eqref{eq:ChenTeo5paramMetric} to be smooth and non-degenerate for $(x,y) \in \Domain$, it is sufficient that $X>0,Y<0,H>0$ in $\Domain$. Recall that
\begin{align} 
\Domain = \{ (x,y) \in (x_2,x_3)\times (x_1,x_2) \} 
\end{align} 
so that for $(x,y) \in \Domain$ and with \eqref{eq:xiordering}
 we have $X > 0$, $Y < 0$ and hence $F>0$. We have that $\Domain$ lies in the third quadrant in the $(x,y)$ plane, below the diagonal, cf. Figure\ref{fig:domain}. The closure $\bar \Domain$ contains the rods $\Rod_k$ given in \eqref{eq:Rodk5param} as boundary segments and the three turning points $n_i$ given in \eqref{eq:nuts} are located at the corners of $\Domain$, while the corner of $\Domain$ on the diagonal, $(x_2, x_2)$, represents the asymptotic end of the instanton.

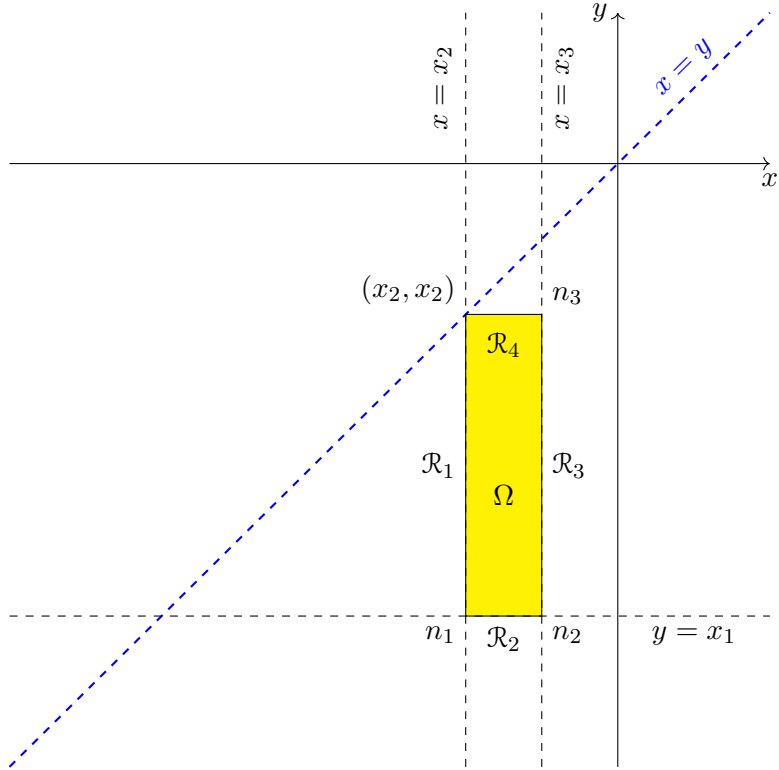
\begin{figure}[t!]
\begin{tikzpicture}[scale=2]

\draw[thin, black,->] (-4,0) -- (1,0)node[anchor=north] {$x$};
\draw[thin,black,->] (0,-4) -- (0,1) node[anchor=east] {$y$};
\draw[thick, blue, dashed](-4,-4) to node[pos=.95,anchor=south east, rotate=45]{$x=y$} (1,1);

\filldraw[fill=yellow,draw=black] (-1,-1) to node[anchor=east]{$\Rod_1$}  (-1,-3) to node[anchor=north]{$\Rod_2$} (-.5,-3) to node[anchor=west]{$\Rod_3$} (-.5,-1) to node[anchor=north,yshift=-1mm]{$\Rod_4$} cycle  ;

\draw (-1,-1) node[anchor=south east] {$(x_2, x_2)$ };
\draw[thin, black, dashed](-4,-3) to node[anchor=north,pos=0.9]{$y = x_1$} (1,-3) ;

\draw (-.75,-2.2) node {$\Domain$};

\draw[thin, black, dashed](-.5,-4) to node[pos=0.9,xshift=3mm,rotate=90]{$x = x_3$} (-.5,1) ;
\draw[thin, black, dashed](-1,-4) to node[pos=0.9,xshift=-3mm,rotate=90]{$x = x_2$} (-1,1) ;

\draw (-1,-3) node[anchor=north east]{$\Nut_1$};
\draw (-.5,-3) node[anchor=north west]{$\Nut_2$};
\draw (-.5,-1) node[anchor=south west]{$\Nut_3$};

\end{tikzpicture} 

\caption{The domain $\Domain$}
\label{fig:domain}
\end{figure}

\begin{lemma} \label{lem:Hg0}
 $H > 0$ in $\bar\Domain$.
\end{lemma}
\begin{proof}
Let
\begin{align} 
L ={}& x\partial_x - y\partial_y , \\ 
A ={}& \frac{2 (\nu x^2 + y^2)}{(x -  y) (\nu x + y)}, \\ 
B={}& - \frac{2 (1 -  \nu) F (\nu x + y)}{(x -  y)^2}.
\end{align} 
Since $y < x$ in $\Domain$, the flow of $L$ is transverse to rods $\Rod_2, \Rod_3$, and the flow lines starting on $ \Rod_2 \cup n_2 \cup \Rod_3 $ foliate $\Domain$. 
We can write equation \eqref{eq:Hflow} in the form 
\begin{align} \label{eq:HflowAB} 
(L - A) H = B.
\end{align} 
This can be solved by introducing an integrating factor. Let $(x_0, y_0)$ be given and let $\hat A$ be the solution to 
\begin{align} 
L \hat A ={}& A, \quad \hat A(x_0, y_0) = 0.
\end{align}
Then can write \eqref{eq:HflowAB} in the form 
\begin{align} 
L(e^{-\hat A} H) = e^{-\hat A} B.
\end{align} 
We find that $e^{-\hat A} H$ is monotone increasing along the flow of $L$, and hence since $e^{-\hat A}> 0$ by construction, we have that $H > 0$ in $\bar\Domain$, provided that $H > 0$ on $\Rod_2$ and $\Rod_3$.  

For $\xi$ satisfying \eqref{eq:xicond}, we have that 
\begin{align}\label{eq:nucond} 
-1 < \nu < -1/2.
\end{align} 
We note that $\nu x + y < 0$ in $\Domain$. Hence, in view of \eqref{CTMetricFunctions:H}, it is sufficient to study 
\begin{align} 
\tilde H = - H/(\nu x + y) 
\end{align} 
on $\Rod_2$ and $\Rod_3$. 
 
\subsubsection*{The case $\Rod_2$}
Setting $y=x_1$, $x_2= -1$ yields $\tilde H_2(x) := \tilde H(x,x_1)$. We have that $\tilde H_2$ is concave, since 
\begin{align} 
\partial_x^2 \tilde H_2 =  2 x_{1}{} \bigl(-2 x_{1}{} + \nu (1 + x_{1}{} -  x_{3}{})\bigr) < 0.
\end{align} 
Hence, in order to show that $\tilde H_2 > 0$ on $\Rod_2$, it is sufficient to show that $\tilde H(\Nut_1) = \tilde H_2(-1) >  0$, and $\tilde H(\Nut_2) = \tilde H_2(x_3)>0$. We have  
\begin{align} 
\tilde H_2(-1) = (\nu -  x_{1}{}) x_{1}{} (1 + x_{1}{})
\end{align} 
and since $x_1 < -1$, we have $\tilde H_2(-1) > 0$. Next we consider 
\begin{align} 
\tilde H_2(x_3) = x_{1}{} (x_{1}{} -  x_{3}{}) x_{3}{} (x_{1}{} + \nu x_{3}{})
\end{align} 
and using \eqref{eq:xiordering} and \eqref{eq:nucond} we find $\tilde H_2(x_3) > 0$. Thus 
\begin{align} 
H \big{|}_{\Rod_2} > 0
\end{align} 
\subsubsection*{The case $\Rod_3$}
Setting $x=x_3$, $x_2=-1$ yields $\tilde H_3(y):=\tilde H(x_3,y)$. We have $\tilde H(\Nut_2) = \tilde H_3(x_1) > 0$ by the discussion for $\Rod_2$. Further, we have 
\begin{align} 
\tilde H(\Nut_3) = \tilde H_3(-1) = x_{3}{} (1 + x_{3}{}) (-1 + \nu x_{3}{}) > 0
\end{align} 
in view of \eqref{eq:xiordering} and \eqref{eq:nucond}. We note that $\tilde H_3$ contains a factor $x_3$ and therefore it is sufficient to consider 
\begin{align} 
\hat H_3(y) = \tilde H_3(y)/(-x_3) = -2 (-1 + \nu) x_{3}{} y^2 + (\nu x_{3}{} -  y) \bigl(x_{1}{} + (-1 + x_{1}{} + x_{3}{}) y\bigr)
\end{align} 
We have 
\begin{align} 
\partial_y^2 \hat H_3 =2 - 2 x_{1}{} + (2 - 4 \nu) x_{3}{}
\end{align} 
We have $\hat H_3$ is concave provided $x_1 > 1 + x_{3}{} - 2 \nu x_{3}{}$, in which case it follows from the fact that $\tilde H(\Nut_2) > 0$, $\tilde H(\Nut_3) > 0$ that $\tilde H > 0$ on $\Rod_3$. 

For the case $x_1 < 1 + x_{3}{} - 2 \nu x_{3}{}$, corresponding approximately to $\xi \in (1/2,0.625)$, $\hat H_3$ is convex. Hence it is sufficient to show that $\partial_y \hat H_3(-1) < 0$, i.e. that $\hat H_3$ is decreasing at $\Nut_3$. We have 
\begin{align} 
\partial_y \hat H_3(-1) ={}& x_{1}{} + \nu x_{1}{} x_{3}{} - 2 (1 + x_{3}{}) + \nu x_{3}{} (3 + x_{3}{}) \\
\intertext{use \eqref{eq:nuxxi} with $x_2=-1$}
={}& - \frac{3 - 14 \xi + 20 \xi^2 - 16 \xi^4 + 24 \xi^5 - 64 \xi^6 + 64 \
\xi^7 - 16 \xi^8}{8 (-1 + \xi)^2 \xi^2 (-1 + 2 \xi)} \label{eq:dyH3}
\end{align} 
The denominator of \eqref{eq:dyH3} is positive, and the numerator can be factorized as 
\begin{align} 
(-1 + 2 \xi^2)P(\xi)
\end{align} 
with 
\begin{align} 
 P(\xi) = 3 - 14 \xi + 26 \xi^2 - 28 \xi^3 + 36 \xi^4 - 32 \xi^5 + 8 \xi^6
\end{align} 
We consider $P(\xi)$ for $\xi \in \xiInt = (1/2,1/\sqrt{2})$. 
We have that $(-1+2\xi^2) < 0$ for $\xi \in \xiInt$. Hence, $\partial_y \hat H_3(-1) < 0$ if and only if $P>0$. 
We have 
\begin{align} 
\partial_\xi^3 P ={}& -168 + 864 \xi - 1920 \xi^2 + 960 \xi^3, \\
\partial_\xi^4 P ={}& 864 - 3840 \xi + 2880 \xi^2
\end{align} 
The quadratic polynomial  $\partial_\xi^4 P$ has two simple roots in the complement of  $\xiInt$ and has a positive second order coefficient. Hence $\partial_\xi^4 P < 0$ in $\xiInt$. We have 
$\partial_\xi^3 P(1/2) = -96< 0$, and thus $\partial_\xi^3 P < 0$ in $\xiInt$. It follows that $\partial_\xi P$ is concave in $\xiInt$. Now, a calculation shows that $\partial_\xi P$ is positive at the endpoints of $\xiInt$. It follows that $\partial_\xi P > 0$ so that $P$ is monotone increasing on $\xiInt$. Finally, $P(1/2)  = 3/8 > 0$, and hence $P> 0$ on $\xiInt$. 

In view of the above discussion, this proves that $\partial_y \hat H_3(-1) < 0$, and hence that $\tilde H > 0$ and also $H > 0$ on $\Rod_3$. In view of the monotonicity of $e^{-\hat A} H$ along the flow of $L$, and the fact that the flow-out of $\Rod_2 \cup n_2 \cup \Rod_3$ is $\Domain$ together with positivity of $H$ at the turning points, this shows that $H > 0$ in $\bar \Domain$. 
\end{proof} 

\begin{remark}
On the boundary $\partial\Omega$ the Weyl tensor is of type $D^+ D^-$ and by \eqref{eq:PsiValues} and \eqref{eq:Psi0tildeRod} the Kretschmann scalar \eqref{eq:Kretschmann} takes the simplified form
\begin{align}
\vert W\vert^2 ={}& 96 (\Psi_{2}^2 + \tilde{\Psi}_{2}^2).
\end{align}
From the explicit form of the curvature scalars in \eqref{eq:PsiValues}, \eqref{eq:Psi2tilde} and Lemma~\ref{lem:Hg0} it follows that $\vert W\vert^2$ is finite and non-vanishing on $\bar \Domain\setminus\{(x_2,x_2)\}$.
\end{remark}

\subsection{Asymptotics} \label{sec:asymptotics}
Provided the rod structure is regular as discussed in Section \ref{sec:prel}, the smooth space $\cM$ is constructed by performing an identification of the Killing coordinates according to \eqref{eq:tauphi-identify}. We emphasize that AF spaces with parameter $\Omega \neq 0$ are not asymptotically Schwarzschild due to the fact the the fundamental domain of the Killing coordinates differs.

In order to put the Chen-Teo instanton in AF form, we introduce new coordinates $(\tilde\tau, \tilde \phi)$ according to 
\begin{align}\label{eq:tphittau}
\phi = \ell_0^\phi \tilde \phi, \qquad
\tau = \ell_0^\tau \tilde \phi+ \sqrt{1-\nu^2} \tilde \tau,
\end{align}
so that $\partial_{\tilde \phi} = \ell_0$ and $\partial_{\tilde \tau}$ has unit norm at infinity.
Following \cite[Eq. (3.7)]{2015PhRvD..91l4005C}, let  
\begin{align} 
x = x_2 - \frac{x_2 \sqrt{k(1-\nu^2)}}{r} \cos^2 \frac{\theta}{2}, \quad y =  x_2 + \frac{x_2\sqrt{k(1-\nu^2)}}{r} \sin^2\frac{\theta}{2} .
\end{align} 
Passing to spheroidal coordinates $(\tilde\tau, r,\theta, \tilde \phi)$ we have 
\begin{align} 
g_{ab} \d x^a \d x^b = \d\tilde\tau^2 + dr^2 + r^2 (\d\theta^2 + \sin^2\theta \d\tilde\phi^2) + O(1/r).
\end{align} 
This shows that the Chen-Teo instanton is AF in the sense of Definition \ref{def:AF} with parameters
\begin{align} \label{eq:param}
\kappa = \ell_0^\phi/\ell_1^\phi , \quad \Omega/\kappa  = \frac{(\ell_1^\tau \ell_0^\phi-\ell_0^\tau\ell_1^\phi)}{\sqrt{1-\nu^2}\ell_0^\phi}.
\end{align} 
See also \cite{BAIRD2021104310} for a discussion using a different parametrization of the instanton.

\subsection{Proof of Theorem \ref{thm:main-intro}} \label{sec:ProofMainThm} 

We have shown in Section \ref{sec:regularity} that the 2-parameter family of Ricci flat metrics given by 
\eqref{eq:ChenTeo5paramMetric} restricted to the parametrization \eqref{eq:nuxxi} define, upon suitable identifications of the Killing coordinates as discussed in Section \ref{sec:rods}, a smooth Ricci flat four-manifold. In view of the discussion in Section \ref{sec:asymptotics} the Chen-Teo instanton is complete and AF, with parameters given by \eqref{eq:param}. Recall that $\nu \neq \pm 1$ for the Chen-Teo instanton, cf. \eqref{eq:nucond}. Taking this into account, it follows from Theorem \ref{thm:AlgNonSpec} the Chen-Teo instanton has algebraically special, but non-vanishing self dual Weyl curvature, and that its anti-self dual Weyl curvature is algebraically general.  In particular, the Chen-Teo instanton is Hermitian and non-K\"ahler, cf. Corollary \ref{cor:details}.

\subsection*{Acknowledgements}
L.A. thanks thanks Shing-Tung Yau for an introduction to gravitational instantons, and many discussions on the subject. We are grateful for the suggestion by Paul Tod, to use a tetrad as in Definition~\ref{def:KillingTetrad}. We also thank Olivier Biquard, Mattias Dahl, Gustav Nilsson, and Walter Simon for many enlightning discussions about instantons.  The main result was established while the authors were in residence at Institut Mittag-Leffler in Djursholm, Sweden during the fall of 2019, supported by the Swedish Research Council under grant no. 2016-06596.


\newcommand{\arxivref}[1]{\href{http://www.arxiv.org/abs/#1}{{arXiv:#1}}}

\newcommand{\prd}{Phys. Rev. D} 
\bibliographystyle{abbrv}

\def\cprime{$'$}

\bigskip
\end{document}